\documentclass[a4paper, 10pt]{article}
\usepackage[T1]{fontenc}
\usepackage[utf8]{inputenc}
\usepackage[english]{babel}
\usepackage{amsmath}
\usepackage{amssymb}
\usepackage{amsfonts}
\usepackage{amsthm}
\usepackage{stmaryrd}

\usepackage{graphicx}
\usepackage{hyperref}
\hypersetup{colorlinks=true}
\usepackage[ruled,vlined]{algorithm2e}

 \newtheorem{theo}{Theorem}[section]
  \newtheorem{theorem}[theo]{Theorem}
 \newtheorem{prop}[theo]{Proposition}
 \newtheorem{proposition}[theo]{Proposition}
 
 \newtheorem{defi}[theo]{Definition}
 \newtheorem{definition}[theo]{Definition}
 
 \newtheorem{lemma}[theo]{Lemma}
\theoremstyle{plain} \newtheorem{notation}[theo]{Notation}
\theoremstyle{plain} \newtheorem{remark}[theo]{Remark}
\theoremstyle{plain} \newtheorem{problem}{Problem}[section]

\newcommand{\K}{\mathbb{K}}

\newcommand{\priv}[1]{\backslash\{ #1 \}}
\newcommand{\intvl}[1]{\llbracket #1 \rrbracket}
\newcommand{\eqdef}{\stackrel{\text{def}}{=}}

\newcommand{\Fq}{\mathbb{F}_q}
\newcommand{\fq}{\Fq}
\newcommand{\Fqm}{\mathbb{F}_{q^m}}
\newcommand{\fqm}{\Fqm}

\newcommand{\CG}[2]{\begin{bmatrix}#1 \\ #2\end{bmatrix}_q}
\newcommand{\OO}[1]{\mathcal{O}\big( #1 \big)}

\newcommand{\C}{\mathcal{C}}

\newcommand{\Cfo}[2]{\overline{#1}^{#2}}
\newcommand{\Cp}[1]{\overline{\C}^{#1}}

\newcommand{\R}{\mathcal{R}}

\newcommand{\word}[1]{\ensuremath{\boldsymbol{#1}}}
\newcommand{\cv}{\word{c}}

\newcommand{\ev}{\word{e}}
\newcommand{\fv}{\word{f}}

\newcommand{\vv}{\word{v}}
\newcommand{\xv}{\word{x}}

\newcommand{\Av}{\word{A}}
\newcommand{\Bv}{\word{B}}
\newcommand{\Cv}{\word{C}}
\newcommand{\Gv}{\word{G}}
\newcommand{\Hv}{\word{H}}
\newcommand{\Mv}{\word{M}}
\newcommand{\Pv}{\word{P}}
\newcommand{\Qv}{\word{Q}}
\newcommand{\Rv}{\word{R}}

\DeclareMathOperator{\Rank}{Rank}
\DeclareMathOperator{\rank}{Rank}

\DeclareMathOperator{\prob}{Prob}

\title{New algorithms for decoding in the rank metric and an 
attack on the LRPC cryptosystem}
\author{%
{Adrien Hauteville{\small $~^{\#*1}$}, Jean-Pierre Tillich{\small $~^{*2}$} }%
\vspace{0.6mm}\\
\fontsize{10}{10}\selectfont\itshape
$^{\#}$Universit\'e de Limoges, XLIM-DMI
123, Av. Albert Thomas, 87060 Limoges, Cedex, France\\
\fontsize{9}{9}\selectfont\ttfamily\upshape
\fontsize{10}{10}\selectfont\rmfamily\itshape
$^{*}$\,Inria, Domaine de Voluceau,
BP 105, Le Chesnay 78153, France\\
\fontsize{9}{9}\selectfont\ttfamily\upshape
$^{1}$\,adrien.hauteville@etu.unilim.fr\\
$^{2}$\,jean-pierre.tillich@inria.fr%
}

\begin{document}
\maketitle
\newpage
\begin{abstract}
We consider the decoding problem or the problem of finding low weight codewords for rank metric codes.
We  show how additional information 
about the codeword we want to find under the form of certain linear combinations of the entries of the
codeword leads to algorithms with a better complexity. This is then used together with a folding technique for attacking 
a McEliece scheme based on LRPC codes. It leads to a feasible attack on one of the parameters
suggested in \cite{GMRZ13}.
\end{abstract}
\section{Introduction}

{\em McEliece schemes.}
The hardness of the problem of decoding a linear code makes its use very attractive in the 
cryptographic setting. Indeed it has been proven to be
NP-complete  for the Hamming metric in the seminal paper of Berlekamp, McEliece and van Tilborg \cite{BMT78}. Moreover, despite some significant research efforts, only 
exponential algorithms  are known for it and the exponent has decreased only 
very slowly over time \cite{BJMM12}. One 
of the very first public-key cryptosystem   \cite{M78} is actually (partly) based on  this problem.
 It still belongs to the very few public key cryptosystems which remain 
unbroken today. 

One of the drawbacks of this scheme is its large public key size. It relies on a particular code family, 
namely Goppa codes,
which in many respects look like  random linear codes but still have an efficient decoding algorithm.
 Since then, many approaches have been tried
to reduce the key size: (i) alternative code families have been proposed,
(ii) using codes with a large automorphism group such as quasi-cyclic codes,
(iii) changing the metric used for the code and the code itself. 

{\em McEliece schemes based on rank metric codes.} In this paper we focus on a proposal which is a mixture of the approaches (ii) and (iii): the LRPC scheme of
\cite{GMRZ13}. It relies on a new family of codes, called Low Rank Parity Check (LRPC in short) codes
which are devised  for the rank metric. The first McEliece scheme based on rank metric codes 
was the Gabidulin-Paramonov-Tretjakov cryptosystem \cite{GPT91}. It relies on an analogue of Reed-Solomon codes for the 
rank metric, the  ``Gabidulin codes''.
The scheme got broken by Overbeck in \cite{O05}. One of the main reasons for 
its insecurity
can be traced back to its  rich algebraic structure. This is not the case for the LRPC scheme. 
 For this family of codes, like for the MDPC codes based McEliece scheme of \cite{MTSB13},
it seems that key security and message security really rely on the same problem, namely 
finding a low rank weight (or moderate Hamming weight for \cite{MTSB13}) codeword  in a linear code
with no structure.

{\em Decoding for the rank metric.} It is essential with this approach to have a good assessment of the complexity 
of solving the decoding problem in the rank metric. Recall that in Delsarte's language \cite{D78}, linear rank metric codes are viewed as the subspace generated by a set of matrices of a same size over some finite field
$\fq$. 

The associated decoding problem is also known 
 under the name "MinRank" and is known to be NP-complete \cite{BFS99}.
 Generally such codes arise in the form of linear codes defined over some extension field $\fqm$.

This problem has attracted some attention in the cryptographic community and algorithms of
exponential complexity have been devised for it \cite{CS96,OJ02,GRS13,KS99,FLP08}.
 
{\em Attacks on quasi-cyclic codes by folding the code.} The parameters of the LRPC scheme have been devised in order to be 
safe against the aforementioned algorithms for decoding in the rank metric. However, the authors
of the scheme have also used quasi-cyclic versions of such codes in order to reduce further the 
size of the parameters. It has been found out recently \cite{FOPPT14a,FOPPT15} that McEliece versions
based on quasi-cylic or quasi-monoidic codes can be attacked by reducing the size of the code
by adding coordinates which belong to the same orbit of the automorphism group. This is called 
the ``folding'' process in these papers. When this process is applied to 
quasi-cyclic or quasi-dyadic alternant or Goppa codes suggested in the cryptographic
community, this results in a much smaller alternant or Goppa code and this can be used to mount a key recovery attack.  This approach was further investigated and the 
folding process was generalized by using a polynomial formalism in \cite{L14}. It was shown there that this 
approach can be used for the quasi-cyclic  LRPC codes of \cite{GMRZ13} and gives  a LRPC
code of much smaller size but which still has in its dual low weight codewords. The decoding algorithm
 of \cite{GRS13} can then be used to find these low weight codewords in a more efficient way than for the
original code. This results in a multiplicative gain in the complexity of the attack of order $2^{12}$ for 
one of the parameters proposed in \cite{GMRZ13}.

{\em Our contribution.} Our contribution in the paper is threefold. First we show how certain rank decoding algorithms   of \cite{OJ02,GRS13} may benefit from some partial knowledge on the codeword which is sought.
We consider here that we are given certain linear combinations of the entries of the
codeword. This generalizes the $\fqm$ linear case where a certain entry can be assumed to be equal to $1$.
Roughly speaking, when we search in the latter case for a rank weight $w$ codeword using the algorithm 
of \cite{OJ02,GRS13} we have algorithms of complexity $q^{(w-1) \alpha}$ where $\alpha$ is some
quantity that depends on the algorithm which is considered and some code parameters. We show how 
the complexity of these algorithms can be reduced to  $q^{(w-a) \alpha}$ when we know $a$ independent
linear combinations of the code positions. We also obtain by the approach of \cite{GRS13} applied to the 
transposed code an algorithm with the same complexity as \cite{OJ02} but which 
is  significantly simpler. Finally, we show that when the folding process
is applied to the quasi-cyclic $\fqm$ linear codes considered in \cite{GMRZ13} we know two independent linear combinations 
of the codeword we are looking for, instead of just one. This is then used together with the generalized folding 
process of \cite{L14} to give a much more efficient attack than  in \cite{L14}.

\section{Generalities about rank metric codes}\label{sec:generalities}

Let us start with the definition of a matrix code
\begin{definition}[Matrix code]
A matrix code of size $m\times n$ over $\Fq$ is a linear code generated by matrices of size $m \times n$ over
$\Fq$. When the code is of dimension $K$ we say that it is an $[m \times n,K]$ matrix code over $\Fq$.
\end{definition}

\begin{remark}
It will be convenient  to express $K$ under the form $K=k.m$. Notice that 
$k$ is not necessarily an integer.
\end{remark}

It might be thought that this is nothing but a linear code of length $m.n$. The point of this definition is that we
equip such codes with the rank metric that is defined by
 $d(\Av,\Bv) = \Rank(\Av-\Bv)$. The weight $|\cv|$ of a word $\cv$ is taken with respect to the rank, that is $|\cv|
\eqdef d(\cv,0) = \Rank(\cv)$.
Generally such codes are obtained from $\Fqm$ linear codes as follows
\begin{definition}[Matrix code associated to an $\Fqm$ linear code]
Let $\C$ be an $[n,k]$ linear code over $\Fqm$ and let $(\beta_1\dots \beta_m)$ be a basis of $\Fqm$ over $\Fq$. Each word $\cv \in \C$ can be represented by an $m\times n$ matrix $\Mv(\cv) = (M_{ij})_{\substack{1 \leq i \leq m\\ 1 \leq j \leq n}}$ over $\Fq$, with $c_j = \sum_{i=1}^m M_{ij} \beta_i$.
The set $\{\Mv(\cv),\cv \in \C\}$ is  the matrix code associated to the $\Fqm$ linear code $\C$.
It is of type $[m\times n,k.m]$.
\end{definition}

This definition depends of course of the basis chosen for $\Fqm$. However changing the basis does not change 
the distance between codewords. The point of defining matrix codes in this way is that they have a more compact 
description. It is readily seen that an $[m\times n,k.m]$ matrix code can be specified from a systematic generator matrix by  
$k(n-k)m^2 \log_2 q$ bits
whereas a  $\Fqm$-linear code uses only $k(n-k)\log_2 q^m=k(n-k)m \log_2 q$ bits. This is particularly interesting for cryptographic applications where this notion is directly related to the
public key size. 
We can now define  the two central problem in this field, namely
\begin{problem}[Decoding in the rank metric] For a given matrix code $\C$ of type $[m \times n,K]$ over 
$\Fq$, a matrix $\Av$ in $\Fq^{m \times n}$ and an integer $w$, find a codeword  $\cv$ in $\C$ such that
$\Rank(\Av - \cv)=w$.
\end{problem}

\begin{problem}[Low rank codeword problem]
For a given matrix code $\C$ and an integer $w$,
find a codeword $\cv$ of rank weight $w$ in  $\C$.
\end{problem}

The decoding problem  reduces to the low rank codeword problem by finding a codeword of weight $w$ in the matrix code
$\C'$ where $\C'$ is generated by the codewords of $\Cv$ and $\Av$ when $\C$ does not contain
codewords of rank weight $w$. In other words, decoding an error of weight $w$ in an $[m \times n,K]$ matrix code
reduces to the problem of finding a codeword of weight $w$ in an 
$[m \times n,K+1]$ matrix code. Notice that the low rank codeword problem is slightly simpler for 
matrix codes obtained from $\Fqm$ linear codes.  Indeed, we may assume that the 
codeword $\cv$ of weight $w$ contains a coordinate equal to $1$. This follows from the fact that  multiplying $\cv$ by any 
nonzero element of $\Fqm$ does not change the rank of the associated matrix. In other words, we 
have some additional knowledge about the codeword (or the error) of weight $w$ in this case. Notice that the support trapping decoding algorithm
(see next section) of \cite{GRS13} and the decoding algorithm of \cite{OJ02} given for $\Fqm$ both 
exploit this knowledge. They have an asymptotic exponential complexity of the form 
$q^{\alpha(w-1)}$ whereas it would have been only $q^{\alpha w}$ for an unstructured 
matrix code with the same parameters.

\section{A support trapping decoding algorithm }\label{sec:basics}

\cite{GRS13} has introduced a very neat and simple algorithm for decoding in the rank metric. It can be considered
as a support trapping decoding algorithm for an $[m \times n,k.m]$
 matrix code that tries to guess a subspace $F$ of the column space $\Fq^m$ of $m \times n$ 
 matrices over $\Fq$ that contains the column space $E$ of the error $\ev$ we want to find.
 Since we focus on the low-weight finding problem in this article we will explain 
 this algorithm in the case we look for a codeword of weight $w$ in an $[m \times n,K]$ code.
 In this case, $E$ is the column space of $\cv$.
 The next step is then to express the columns $\cv_i$ of $\cv$ in a basis $\fv_1,\dots, \fv_r$ of $F$, that is
$
 \cv_i = \sum_{j=1}^r x_{ij} \fv_j.
 $
 This gives $n.r$ unknowns (the $x_{ij}$'s). From a parity-check matrix of the matrix code we deduce
 $n.m-k.m=(n-k)m$ equations involving the entries of $\cv$ that can all be expressed in terms of the $x_{ij}$'s.
 In other words, we have a linear system with $(n-k)m$ equations and $n.r$ unknowns. We choose
 $r$ to be  the least integer such that the number of unknowns is less than the number of equations. In our case, $r = m - \big \lceil \frac{km}{n} \big\rceil$.

The complexity of the algorithm depends on the probability of having $E \subset F$. It is equal to the number of subspaces of dimension $w$ in a subspace of dimension $r$, divided by the number of all subspaces of dimension $w$ in $\Fq^m$. This probability can be easily expressed with Gaussian coefficients, which counts the number of subspaces of a vector-space :
\begin{equation}\label{eq:proba}
 p = \frac{\CG{r}{w}}{\CG{m}{w}} = \Theta\left( q^{-w(m-r)} \right)
\end{equation}
We use here the following notation.
\begin{notation}
$\CG{m}{w}$ is the Gaussian binomial coefficient that is equal to the number of 
subspaces of $\Fq^m$ of dimension $w$. Recall that this coefficient satisfies
$\CG{m}{w} = \Theta\left(q^{w(m-w)} \right)$.
\end{notation}

The cost to solve a linear system of $(n-k)m$ unknown by Gaussian elimination is $\OO{(n-k)^3 m^3}$. Thus, the overall expected complexity for this algorithm is $\OO{(n-k)^3 m^3 q^{w \lceil \frac{km}{n} \rceil}}$.

As explained in  \cite{GRS13} $F$ can be viewed as the support of a codeword 
for the rank metric.  What makes this notion interesting is that it establishes a parallel with the Hamming metric : indeed, if we know the support $\cv$ of a codeword $\cv$  we can recover $\cv$  in polynomial time by solving a linear system.

This algorithm is much more efficient than the algorithm in \cite{OJ02} when $m \leq n$. Let us notice 
that in the case $m >n$ we can improve this algorithm in a simple way by using the notion of the transposed
code which is defined as follows  \cite{GP06} 
\begin{definition}[transposed code]
The transposed of an $[m \times n,K]$ matrix code $\C$ over $\Fq$ is a
$[n\times m,K]$ matrix code $\C^T$ over $\Fq$ obtained by
$
\C^T =\{\Mv^T, \Mv \in \C\}.
$
\end{definition}

The idea underlying the definition of such a   code is that  transposing a matrix preserves its rank, therefore 
finding the minimum rank weight (nonzero) codeword $\C$ can be obtained from the transpose of the
minimum rank weight (nonzero) codeword of $\C^T$. Notice that taking the transpose basically swaps
the role of $n$ and $m$. This notion can be used when $m \geq n$ for finding a
codeword  of weight $w$ in a matrix code $\C$ by looking for a codeword of weight $w$ in 
$\C^T$. It is readily seen that this leads to an 

algorithm of complexity $\OO{(n-k)^3 m^3 q^{w \lceil k \rceil}}$ for finding a codeword of weight $w$.
This is precisely the complexity that the algorithm of \cite{OJ02} would give for finding 
a codeword of weight $w$ in a matrix code. However the algorithm presented here is much simpler
than the algorithm of \cite{OJ02}.

\section{A  low weight codeword finding algorithm 
using additional knowledge on the codeword}
\label{sec:transposed}

In this section, we  assume that we have additional knowledge about the codeword of
weight $w$ we want to find in the form of linear combinations of its columns.
More precisely we are looking for an algorithm whose input and output are specified 
in Algorithm \ref{algo:transposed}.
\begin{algorithm}
\SetKwInOut{Input}{Input}\SetKwInOut{Output}{Output}
\SetKwInOut{Assume}{Assumes}
\Input{{}\\
\noindent
(i) an 
 $[m \times n,k.m]$ matrix code $\C$ over $\Fq$ that has at least one codeword $\cv = (c_{ij})_{\substack{ 1 \leq i \leq m\\ 1 \leq j \leq n }}$ of rank weight 
$w$ \\
\noindent
(ii)  $a$ elements $\cv'_1,\dots,\cv'_a$ in $\Fq^m$ that are linear combinations of  columns of $\cv$.\\
\noindent
(iii) the coefficients $\lambda_{ij}$'s of these linear combinations, that is if we denote  by 
$\cv_{.,j}= (c_{ij})_{1 \leq i \leq m}$ the $j$-th column of $\cv$,  then  
$\cv'_i = \sum_{j=1}^n \lambda_{ij} \cv_{.,j}$ for $i \in \{1,\dots,a\}$. }
\Assume{ $\cv'_1,\dots,\cv'_a$ are linearly independent.}
\Output{a codeword of $\C$ of rank weight $w$.}
\caption{Low rank codeword finding with additional information
\label{algo:transposed}}
\end{algorithm}

The case of a matrix code obtained from an $\fqm$-linear code is a particular case of such 
an additional knowledge: as explained before we can assume that one of the columns of
the codeword we are looking for is the column $\left(10 \dots 0\right)^T$.
The folding attack that we present in Section \ref{sec:attack} will provide another example where
we have the knowledge of two independent linear combinations of the columns and will use in an 
essential way the algorithm we give here.

\subsection{The case $n \geq m$}\label{ss:first}

We use here a variation of the support trapping algorithm \cite{GRS13}.
The case when $a=1$ and  when the matrix code is obtained from an $\fqm$-linear code is already 
treated in \cite[Prop. 3.1]{GRS13}. Generalizing this argument to the more general setting 
considered here just consists in chosing in the error trapping algorithm recalled in 
Section \ref{sec:basics} an  $F$ as a random subspace of dimension $r$ that contains the subspace generated by the $a$ elements $\cv'_1,\dots,\cv'_a$. This leads to the following proposition.

\begin{proposition}\label{prop:complexity1}
The support trapping algorithm outlined above has expected complexity 
$\OO{(n-k)^3 m^3 q^{(w-a) \lceil \frac{km}{n} \rceil}}$ when applied on a 
matrix code over $\fq$ of type $[m \times n,k.m]$.
\end{proposition}

The complexity given follows almost immediately from the following proposition.

\begin{prop}\label{prop:subspace}
Let $E$ be a subspace of dimension $w$ of $\Fq^m$ and let $E'$ be a subspace of $E$ of dimension $a$. Let $S$ be the set of subspaces of dimension $r$ of $\Fq^m$ that contain $E'$ and
let $F$ be an element of $S$ chosen uniformly at random. We have
$$
\prob \left( E \subset F \right)  = \dfrac{\CG{r-a}{w-a}}{\CG{m-a}{w-a}} = \Theta\left(q^{(w-a)(m-r)}\right).
$$
\end{prop}

\begin{proof}
Let $V = \Fq^m/E' \simeq \Fq^{m-a}$.\\
Let $\pi$ be the canonical surjection from $\Fq^m$ to $V$ :
\[ \left.\begin{array}{rccc}
\pi : & \Fq^m & \rightarrow & V\\
 & \xv & \mapsto & \xv + E'
\end{array}\right.  \]
It is well known that $\pi$ gives a one-to-one correspondence between the subspaces of $\Fq^m$ which contain $E'$ and the subspaces of $V$.\\
\begin{lemma}
Let $F$ be a subspace of dimension $r$ of $\Fq^m$ that contains $E'$. The dimension of  $\pi(F)$ is $r-a$.
\end{lemma}
\begin{proof}
Let $(\ev_1, \dots , \ev_a)$ be a basis of $E'$. We can complete this basis into a basis $(\ev_1, \dots , \ev_a, \fv_1, \dots, \fv_{r-a})$ of $F$. It is obvious that $(\pi(\fv_1), \dots, \pi(\fv_{r-a}))$ is a basis of $\pi(F)$, so $\dim \pi(F) = r-a$
\end{proof}

By using this lemma, we finish the proof of Proposition \ref{prop:subspace}.
There are $\CG{m-a}{w-a}$ subspaces of $V$ of dimension $w-a$ .
This implies that there exist $\CG{m-a}{w-a}$ subspaces of dimension $w$ of $\Fq^m$  that contain $E'$.\\

Let $F$ be a subspace of dimension $r$ of $\Fq^m$ that contains $E'$. According to the previous lemma, $\dim \pi(F) = r-a$. So $\pi(F)$ contains $\CG{r-a}{w-a}$ subspaces of dimension $w-a$. From this, we deduce that
$F$ contains $\CG{r-a}{w-a}$ subspaces of dimension $w$ that contain $E'$.
 Hence 
\[\prob \left( E \subset F \right)  = \dfrac{\CG{r-a}{w-a}}{\CG{m-a}{w-a}} = \Theta\left( q^{-(w-a)(m-r)} \right)\]
\end{proof}

The proof of Proposition \ref{prop:complexity1} follows directly from this Proposition. Indeed
we choose in the support trapping algorithm, $E'$ to be the linear space generated by $\cv'_1,\dots,\cv'_a$ and 
$F$ as a random subspace of $\fq^m$ that contains $E'$. The expected complexity of the support trapping algorithm
is now given by the inverse of the probability that we computed in Proposition \ref{prop:complexity1} multiplied 
by the complexity of solving a linear system with $(n-k)m$ equations.

\subsection{The case $m > n$}\label{ss:second}

This will be treated essentially by a variation on the error trapping algorithm applied to the transposed code which uses in a suitable way the additional knowledge about the codeword we want to find. 
The technical difficulty we face here can be described as follows.
If we had additional knowledge about $\cv$ in the form of $a$ independent elements belonging to the row space
of $\cv$, then we could immediately apply the algorithm given in Section \ref{sec:basics} to the
transposed code. However it turns out that in the case we are interested in, the knowledge about $\cv$ that we  have concerns the column space of $\cv$. In this case, when 
we transpose $\cv$ to reverse the role of $n$ and $m$, this translates into some knowledge of the
row space of $\cv^T$ and we can not use the algorithm of Section \ref{sec:basics} anymore.
This is why we are going to consider a slightly more complicated algorithm which is able to use some
knowledge on the column space of $\cv$.
It will be essential for our attack that is given in Section \ref{sec:attack} to work to 
have an efficient algorithm for finding low-rank codewords by exploiting
some knowledge about the low-rank codeword we are looking for. Even if the underlying code
is defined for $m<n$ it turns out that we are reducing this problem to another
low-rank finding problem in a new code where $m>n$. Of course we could
still use the algorithm described in Section \ref{sec:basics}. It appears that 
 the Ourivski-Johansson \cite{OJ02} is better in the regime when $n < m$.
However, this algorithm in its \cite{OJ02} form is unable to take full advantage
of the knowledge we have about the low-weight codeword we are looking for.
 It would have been possible  
to give a version of the Ourivski-Johansson that exploits additional 
knowledge in the same way we generalized slightly the 
support trapping algorithm of \cite{GRS13}. However, the Ourivski-Johannson algorithm
is rather involved and we will use another approach here that recovers
the same work factor as the Ourivski-Johannsson algorithm in the
case of decoding a $\Fqm$-linear code but in a much simpler fashion.
This new decoding algorithm is in essence a support trapping 
algorithm working on the transposed code. 
  It will also be able to use
in a simple way additional knowledge about the low rank word we are looking for.

 The point is now that by applying a  version of the support trapping algorithm 
of \cite{GRS13} that makes use in a suitable way of the additional knowledge we have about the 
support,
we basically recover  an algorithm with the same complexity as the Ourivski-Johansson algorithm for decoding
$\Fqm$ linear codes. More generally it will have an exponential asymptotic complexity of order 
$ \OO{(n-k)^3m^3 q^{(w-a)k}}$ for an $[m \times n,k.m]$ matrix code over $\Fq$ when we know $a$ independent linear 
combinations of the columns of the matrix codeword of rank $w$ we are looking for.

This algorithm can be described as follows\\
{\bf Step 1 (transformation of the code):} We first transform the matrix code $\C$ by multiplying it at the right by an $n \times n$ invertible
matrix $\Pv$ such that $\cv$ gets transformed in a matrix $\cv'$ whose $i$ first  columns are precisely
the $\cv'_i$'s defined before. In other words, we consider the code 
$\C' = \C \Pv$. If $\cv$ is a word of rank weight $w$ then $\cv'$ is still a word of rank weight
$w$. Moreover by assumption on the independence of the $\cv'_i$'s for $i \in \{1,\dots,a\}$ we can further 
multiply $\C'$ on the left by an $m \times m$  invertible  
matrix $\Qv$ such that $\cv'$ gets transformed in a matrix $\cv"$ whose  first $a$ columns are the 
 first $a$ elements $\ev_1,\dots,\ev_a$ of the canonical basis of $\Fq^m$, that is
$\ev_i$ has only zero entries with the exception of the $i$-th entry which is equal 
to $1$.
Let $\C"$ be the resulting code obtained by these operations, that is
$$
\C" = \Qv \C \Pv 
$$
Notice that $\cv"$ has still rank $w$.

\noindent
{\bf Step 2: (setting up the unknowns of the linear system)} We are now basically going to apply a variation of the support trapping algorithm of \cite{GRS13} on 
$\C"^T$ by choosing a subspace $V$ of $\Fq^n$ of dimension $r$ ($r$ will be specified later on)
for which we hope that it contains the subspace generated by the columns of $\cv"^T$.
A basis $\vv_1,\dots,\vv_r$ of this space is chosen such that
\begin{eqnarray}
v_{j,i} &= &0 \text{ for $i$, $j$ in $\{1,\dots,a\}$ and $i \neq j$}\label{eq:cond0}\\
v_{i,i} &= &1 \text{ for $i$ in $\{1,\dots,a\}$ }\label{eq:cond1}\\
v_{j,i} & = & 0 \text{ for $i$ in  $\{a+1,\dots,r\}$ and $j$ in $\{1,\dots,a\}$}\label{eq:cond2}
\end{eqnarray}
where $v_{j,i}$ denotes the $j$-th coordinate of $\vv_i$.
The entries $v_{j,i}$ are chosen uniformly at random for  $i$ in  $\{a+1,\dots,r\}$ and $j$ in $\{a+1,\dots,n\}$.
The entries of $v_{j,i}$ for  $i$ in  $\{1,\dots,a\}$ and $j$ in $\{a+1,\dots,n\}$ will be chosen 
afterwards.
Denote by $\Cv_1,\dots,\Cv_m$ the $m$ columns of $\cv"^T$. 
Let us introduce the $x_{s,t}$'s in $\Fq$ that are such that
\begin{equation}
\label{eq:starting_point}
\Cv_s= \sum_{t=1}^r x_{s,t} \vv_t \text{  for $s$ in $\{1,\dots,m\}$.}
\end{equation}
Notice now the following point 
\begin{lemma}\label{lem:unknowns}
For $s >a$ and all $i$ in $\{1,\dots,a\}$ we have 
$x_{s,i}=0$.
If we denote by $C_{i,j}$ the $i$'th element of the $j$-th column $\Cv_j$ of 
$\cv"^T$ then
$C_{i,j}=0$ for all $i,j$ in $\{1,\dots,a\}$ with the exception of the diagonal elements
$C_{i,i}$ that are equal to $1$.
\end{lemma}
\begin{proof}
Denote by $\Rv_1,\dots,\Rv_n$ the $n$ rows of $\cv"^T$.
Notice that 
\begin{equation}\label{eq:Ri}
\Rv_i = \ev_i, \text{ for $i$ in $\{1,\dots,a\}$}
\end{equation}
where the $\ev_i$'s are as before the canonical basis of $\Fq^m$.
This implies directly that $C_{i,j}=0$ for all $i,j$ in $\{1,\dots,a\}$ with the exception of the diagonal elements
$C_{i,i}$ that are equal to $1$.
Moreover, by using \eqref{eq:Ri} together with \eqref{eq:cond0},\eqref{eq:cond1} and \eqref{eq:cond2} we know that $x_{s,i}=0$ for $s >a$ and all $i$ in $ \{1,\dots,a\}$.
\end{proof}

This motivates to define as unknowns the $(m-a)(r-a)+a(n-a)$ quantities $x_{s,t}$ and $C_{i,j}$ 
for $s$ in $\{a+1,\dots,m\}$, $t$ in $\{a+1,\dots,r\}$, $i$ in $\{a+1,\dots,n\}$ and $j$ in
$\{1,\dots,a\}$.

Moreover these unknowns satisfy $nm - km=(n-k)m$ linear equations obtained from 
the fact $\cv"^T$ belongs to $\C"^T$ which is a matrix code of dimension $km$. They can be 
obtained by computing a parity-check matrix of this code, then expressing the linear equations
that the entries of $\cv"^T$ have to satisfy and then replacing these entries by the aforementioned
unknowns by using \eqref{eq:starting_point} and Lemma \ref{lem:unknowns}.
We choose $r$ such that the number of equations, that is $(n-k)m$ is at least equal
to the number of unknowns, that is
$$
(n-k)m \geq (m-a)(r-a)+a(n-a) 
$$
This can be obtained by choosing 
$$
r \eqdef \left\lfloor \frac{m}{m-a}(n-k) + a \frac{m-n}{m-a} \right\rfloor
$$

\noindent{\bf Step 3: (solving the linear system)} The last point just consists in solving the linear system, this 
yields $\cv"^T$ and from this we deduce $\cv"$ and then $\cv$ by 
$$
\cv = \Qv^{-1}\cv"  \Pv^{-1}
$$

The last point to understand is under which condition $V$ contains the  subspace generated by the columns of $\cv"^T$.
This depends on how we specify the entries $v_{j,i}$ for  $i$ in  $\{1,\dots,a\}$ and $j$ in $\{a+1,\dots,n\}$. We choose them such that \eqref{eq:starting_point} is verified for
$s$ in $\{1,\dots,a\}$. This can obviously be done by choosing
\begin{equation}
\label{eq:obvious}
\vv_i = \Cv_i \text{ for } i \in \{1,\dots,a\}
\end{equation}

\begin{lemma}\label{lem:contains}
Let $V$ be chosen by a basis $\vv_1,\dots,\vv_r$ such that 
 its $a$ first elements are given by \eqref{eq:obvious} and as specified in 
Step 2 for the other elements.
Let $W$ be the subspace generated by the columns of $\cv"^T$.
Let $W_0$ be the subspace of $W$ that is formed by the elements whose first $a$ entries are all equal to $0$. 
In the same way, we denote by $V_0$ the subspace that is formed by the elements of $V$ whose
 first $a$ entries are all equal to $0$.
We have $W \subset V$ iff $W_0 \subset V_0$.\end{lemma}

\begin{proof}
It is clear that $W \subset V$ implies $W_0 \subset V_0$.

Now assume that $W_0 \subset V_0$. Notice that $W$ is generated by 
$W_0$ and by the first $a$ columns of $\cv"$, that is $\Cv_1,\dots,\Cv_a$.
Since $V$ is generated by the same first $a$ columns of $\cv"$,  $\Cv_1,\dots,\Cv_a$ and by 
$V_0$ we have that $W \subset V$.
\end{proof}

Putting all these considerations together we obtain that
\begin{theorem}\label{th:comp_transposed}
Let $\C$ be an $[m \times n,k.m]$ matrix code which has at least one codeword of rank weight 
$w$  for which we know $a$ independent linear combinations of its columns as specified in 
Algorithm \ref{algo:transposed}.
Assume that $n \leq m$ and
let $r \eqdef \left\lfloor \frac{m}{m-a}(n-k) + a \frac{m-n}{m-a} \right\rfloor$. Then the algorithm 
given in this section outputs a codeword of weight $w$
with complexity 
$\OO{(n-k)^3 m^3 q^{(w-a)(n-r)}}$.
\end{theorem}
\begin{proof}
This follows immediately from Lemma \ref{lem:contains} and Proposition \ref{prop:subspace}
that show that we will try an expected number of $\dfrac{\CG{n-a}{w-a}}{\CG{r-a}{w-a}}=\Theta\left( q^{(w-a)(n-r)}\right)$
spaces $V$ before finding the right one if there is only one codeword $\cv$ which has the  right form.
This is of course an upper bound if there are more than one codeword that have the right form.
Each try of a tentative space $V$ takes time $\OO{(n-k)^3 m^3}$ whose complexity is dominated by Step 3 when we solve
a linear system with $(n-k)m$ equations and a number of unknowns that is less than the number of
equations.
\end{proof}

\section{Folding and projecting attack}\label{sec:attack}
In this section we present a key recovery attack on the LRPC cryptosystem \cite{GMRZ13}.  The codes used there are defined by

\begin{defi}[LRPC code]
Let $C$ be the matrix code associated to the $\fqm$-linear code with a full rank 
parity check matrix $\Hv$ of size $(n-k) \times n$. It defines 
an $[n,k]$ LRPC  of weight $d$ if the $\fq$ subspace of $\fqm$ generated by the entries 
of $\Hv$ is of dimension $d$.
\end{defi}

A probabilistic decoding algorithm of polynomial time for LRPC codes is presented in \cite{GMRZ13}.
This algorithm uses in an essential way that such a code has a parity-check matrix $\Hv$ that has entries in a subspace of small dimension.
$\Hv$ can be easily hidden by giving a systematic parity-check matrix $\Hv_{\text{syst}}$. This family of codes can then be used in a McEliece type scheme \cite{M78} : the secret key is $\Hv$ and the public
key is $\Hv_{\text{syst}}$. To recover the secret key, the attacker must find a word of weight $d$ in the dual of $\C$, which is hard in principle. To decrease the key sizes, double-circulant LRPC codes 
are suggested in \cite{GMRZ13}.

\begin{defi}
A double-circulant LRPC (DC-LRPC) code of weight $d$ over $\fqm$ is an LRPC code defined from a double-circulant parity check matrix $\Hv = \begin{pmatrix} \Hv_1 & \Hv_2 \end{pmatrix}$ where $\Hv_1$ and $\Hv_2$ are two 
circulant matrices and the $\fq$-subspace of $\fqm$ generated by the entries of $\Hv$ is of dimension $d$.
\end{defi}

\subsection{Folded and projected codes}
We  present here two new ingredients of the attacks that follow, namely the notion
of folded code and  the notion of projected code. The first attack uses only folding but the
second attack uses both. The notion of projected codes uses the polynomial framework
for dealing with quasi-cyclic codes \cite{LS01a,LF01}. Quasi-cyclic 
codes are a generalization of double-circulant codes : they are defined by a parity check matrix
formed only from circulant blocks. Such a quasi-cyclic
code of length $N = \ell n$ defined over a finite field $\K$,  where the size of the circulant blocks is $n$, can also be viewed as code over the ring $\K[X]/(X^n-1)$. This is a specific instance of cellular codes that are codes
defined over a ring $\R=\K[X]/(f(X))$ where $f(X)$  is a polynomial of $\K[X]$.

For the reader's convenience, we recall here the polynomial formalism of \cite{LS01,LF01} and follow the 
presentation given in \cite{L14}.
Recall that the fact that quasi-cyclic codes can be viewed as codes defined over the 
ring $\K[X]/(f(x))$ follows directly from 
\begin{prop}\label{prop1}The set of  circulant matrices of size $n \times n$ over 
$\fqm$ is isomorphic to the $\Fqm$-algebra  $\Fqm[X]/(X^n-1)$ by the function $\phi$ 
\[ \phi\left(\sum_{i=0}^{n-1} a_i X^i\right) = \begin{pmatrix}
a_0 & a_1 &\dots & a_{n-1}\\
a_{n-1} & a_0 &\dots & a_{n-2}\\
& \ddots & \ddots & \\
a_1 & a_2 &\dots & a_0
\end{pmatrix} \]
\end{prop}

More generally we consider codes over a finite field $\K$ derived from codes defined over a
ring 
$$
\R \eqdef \K[X]/(f(X))
$$
where $f$ is some polynomial in $\K[X]$ of degree $n$.
They are derived from the following $\K$-isomorphism $\psi:\R \rightarrow \K^{n}$  :
$$
a(X) = \sum_{i=0}^{n-1} a_i X^i \mapsto \psi(a(X)) =(a_0,\dots,a_{n-1}).
$$
They are called cellular codes and are defined by
\begin{definition}[cellular code]
Consider a submodule $M$ of $\R^\ell$ of rank $s$. Let $\psi^\ell: \R^\ell \rightarrow \K^{\ell n}$ that maps an element $(f_1, \dots,f_l)$ of 
$\R^\ell$ to $\K^{\ell n}$ by mapping each $f_i$ to $\psi(f_i)$.
The cellular code associated to $M$ is given by $\psi^\ell(M)$.
It is said to have index $\ell$ and it is a $\K$-linear code of length
$\ell n$. 
\end{definition}

\begin{remark}
In order to avoid cumbersome notation, we identified $M$ with $\psi^\ell(M)$ 
in Section \ref{sec:attack}. Sometimes 
it will better to view the cellular code $\psi^\ell(M)$ as $M$ and we will freely do this.
\end{remark}

To obtain a generator matrix of the cellular code from a generator matrix
$$
\Gv_M = \begin{pmatrix} a_{1,1}(X) &\dots &a_{1,\ell}(X) \\
\vdots &\ddots & \vdots \\
a_{s,1}(X) &\dots &a_{s,\ell}(X)
\end{pmatrix}
$$
of the rank $s$-submodule we introduce the 
following mapping $\phi:\R \rightarrow \K^{n \times n}$:
$$
a(X) = \sum_{i=0}^{n-1} a_i X^i \mapsto \phi(a(X)) =\left( \begin{array}{c} 
\psi(a(X))\\
\psi(Xa(X))\\
\vdots \\
\psi(X^{n-1}a(X))\end{array}\right)
$$
It is a bijective morphism of $\K$-algebras. When $f(X)=X^n-1$ this is  
precisely the mapping that appears in Proposition \ref{prop1}.
A generator matrix of the associated cellular code is now given by
$$
\Gv = \begin{pmatrix} \Av_{1,1} &\dots & \Av_{1,\ell} \\
\vdots &\ddots & \vdots \\
\Av_{s,1} &\dots &\Av_{s,\ell}
\end{pmatrix}
$$
where $\Av_{i,j} = \phi(a_{i,j}(X))$.
This implies that the dimension $k$ of the cellular code satisfies $k \leq n s$.


\begin{definition}[projected code]
Consider a cellular code $\C$ of index $\ell$ defined over $\R \eqdef \K[X]/(f(X))$ and let $g(X)$ be a divisor of
$f(X)$ in $\K[X]$. The projected cellular code $\Cp{g}$ is obtained by viewing a codeword $\cv$
of $\C$ as an element of $R^l$ : $\cv=(\cv_1,\dots,\cv_\ell)$ and applying the surjective morphism 
$\Pi$ from $\K[X]/(f(X))$ to $\K[X]/(g(X))$ defined by $\Pi(a(X)) = a(X) \pmod{g(X)}$ to every entry
$\cv_i$.
\end{definition}

In the particular case where $f(X) = X^n-1$ and $g(X)=X^{m}-1$ where $m$ is a divisor of $n$, projecting
corresponds to folding in the sense of \cite{FOPPT14a,FOPPT15}.
\begin{definition}[folded code]\label{def:folding}
Consider a quasi-cyclic code $\C$ of index $\ell$ and length $n \ell$. Let $m$ be a divisor of $n$.
Its folded code of order $m$ is a quasi-cyclic code of index $\ell$ and length $m \ell$ 
obtained by mapping each codeword $\cv=(c_0,\dots,c_{n \ell -1})$ of $\C$ to the 
codeword $\cv' = (c'_0,\dots,c'_{m \ell -1})$ where 
\[
c'_i = \sum_{s=0}^{\frac{n}{m}-1} c_{an+b+sm}
\]
and $a$ and $b$ are the quotient and the remainder of the euclidean division of $i$ by $m$:
\[
i=am+b \text{ with }a \text{ and }b\text{ integer and }b \in \{0,\dots,m-1\}.
\]
\end{definition}
This really amounts to sum the coordinates that belong to the same orbit of a (permutation) autorphism of
order $n/m$ that leaves the quasi-cyclic code invariant. 

There are two points which make these two notions very interesting in the cryptographic setting.
The first point is that these two reductions of the code do not lead to a trivial code at the
end (one could have feared to end up with the full space after projecting or folding).
This comes from the following proposition that is proved in \cite[Corollary 1]{L14}
\begin{proposition}
Consider a cellular code $\C$ of index $\ell$ defined over $\R \eqdef \K[X]/(f(X))$ and let $g(X)$ be a divisor of
$f(X)$ in $\K[X]$. The length of the projected code $\Cp{g}$ is $\ell \deg g$ whereas the dimension of 
$\Cp{g}$ is less than or equal to $s \ell$ where $s$ is the rank of the cellular code.
\end{proposition}

The second point is that this operation of folding behaves nicely with respect to projecting
a quasi-cyclic code defined over an extension field $\fqm$ with respect to the rank distance
over $\fq$ when the divisor $g$ belongs to $\fq[X]$ 
\begin{proposition}[{\cite[Prop.3]{L14}}]\label{prop:projecting_rank}
Consider a cellular code $\C$ of index $\ell$ defined over $\R \eqdef \fqm[X]/(f(X))$ and let $g(X)$ be a divisor of
$f(X)$ in $\fq(X)$. Denote by $\Pi$ the associated projection operation. We have
$$
\rank(\Pi(\cv)) \leq \rank(\cv)
$$
for any $\cv \in \C$ where we view these codewords  as matrices 
in $\fq^{m \times \ell \deg f}$ or in $\fq^{m \times \ell \deg g}$ by taking the matrix form of
these codewords as defined in 
Section \ref{sec:generalities}.
\end{proposition}

Notice that this proposition can always be applied to folded codes. These two propositions
allow to search for a codeword $\cv$ of rank $w$ in a quasi-cyclic code $\C$ of index $\ell$ and length $n \ell$ defined over $\fqm$ by projecting 
it with respect to a divisor of $X^n-1$ that belongs to $\fq[X]$ (or by folding it) and 
looking for a word of rank $ \leq w$ in the projected or folded code. Roughly speaking, the first proposition 
ensures  that we are not looking 
for a word $w$ in the entire space. From the second proposition, we expect that 
as long $w$ is below the Gilbert-Varshamov bound of the folded code, the codeword
of weight $\leq w$ we will find in the projected code corresponds to the projection
of $\cv$. This allows to recover easily $\cv$.

\subsection{A first attack based on folding}
Let $\C$ be a DC-LRPC $[2k,k]$ code  of weight $d$ over $\Fqm$ obtained from a parity-check matrix 
$\Hv$.
To recover $\Hv$ it is clearly sufficient to find a codeword of rank weight $d$
in the dual $\C^\perp$ of $\C$. 
Let $\C'$ be the folding of order $1$ of $\C^\perp$. 

It is in general a $[2,1]$ code.
This folding reveals some additional information about 
the subspace $F$ of $\Fqm$ generated by the coefficients of $H$.
We namely have

\begin{proposition}\label{prop:attack1}
Let $\cv'=(c'_1,c'_2)$ be in $\C'$.
There exists  $\cv$ of weight $d$ in $\C^\perp$ such that the $\fq$-subspace
generated by the coordinates of $\cv$ contains $c'_1$ and $c'_2$.
\end{proposition}

\begin{proof}
If $\C'$ is the all-zero code or $\cv'=0$ the conclusion follows directly. 

Assume now that this is not the case. In this case, $\C'$ is of dimension $1$.
Consider a codeword $\cv$ of $\C^\perp$ which is of weight $d$. 
Let $\cv"$ be the folded version of $\cv$. We have in this case
$\cv' = \alpha \cv"$ for some $\alpha \in \fqm^\star$.
Note that $\cv'$ is the folded version of 
$\alpha \cv$. We observe now two things and this finishes the proof
\begin{itemize}
\item $d=\Rank(\cv)=\Rank(\alpha \cv)$
where we view these codewords as matrices in $\fq^{m \times n}$ as explained
in Section \ref{sec:generalities}.
\item $c'_1$ and $c'_2$  are in the $\fq$-subspace of $\fqm$ generated by the coordinates of $\alpha \cv$.
\end{itemize}
\end{proof}

 We can use $(c'_1,c'_2)$ in the decoding algorithm described in Section \ref{sec:transposed}. From $\cv$ we recover immediately a parity-check matrix of the form $\beta \Hv$, where $\beta \in \Fqm\priv{0}$, by building a parity-check matrix
from $\cv$ and its cyclic shifts. This gives an attack  of complexity  $\OO{k^3m^3 q^{(d-2)\lceil\frac{m}{2}\rceil}}$. 
 However for  the parameters proposed in \cite{GMRZ13,GRSZ14}, this does not 
 improve the attacks already considered there. However, this proposition together 
 with another projection of the code will lead to a feasible attack against a certain 
 parameter of \cite{GMRZ13,GRSZ14} as we now show.

\subsection{An improved attack based on folding and projecting}
To improve the attack, we search for a word of weight $d$ in a projected code.
This new attack depends  on the factorization of $X^k-1$. The length of the projected code we are interested in will be smaller than $m$ and we will use the algorithm of Subsection \ref{ss:second} instead.
The attack can be described as\\
{\bf Step 1:} 
Compute $\C'$ the folding of order $1$ of $\C^\perp$ and extract a codeword $(c'_1,c'_2)$ in it.\\
{\bf Step 2:} Compute the projected code $\Cfo{\C^\perp}{D}$ with respect to a certain divisor $D(X)$ of $X^k-1$ in 
$\fq[X]$.\\
{\bf Step 3:} Find a codeword $\cv"$ in $\Cfo{\C^\perp}{D}$ of weight $w$ such that the $\fq$ space generated by its coordinates
contains $c'_1$ and $c'_2$ by using the algorithm of Subsection \ref{ss:second}.\\
{\bf Step 4:} Let $F$ be the $\fq$-subspace of $\fqm$ generated by the coordinates of $\cv"$. Find the codeword $\cv$
in $\C^\perp$ of rank weight $w$ whose support is $F$ (meaning that the $\fq$-subspace 
generated by its coordinates should belong to $F$.)

What justifies the third step is the fact that Proposition \ref{prop:attack1} generalizes easily to 
the projected code, whereas what justifies Step 4 is the fact that it is extremely likely 
that $\cv"$ is the projection of a codeword in $\C^\perp$ of weight $d$ we are looking for. 
We recover in this case such a codeword by the process of Step 4. 
The complexity of this attack 
is dominated by the third step and is given by Theorem \ref{th:comp_transposed}.

In  \cite{GMRZ13}, some  parameters for the LRPC cryptosystem are suggested. 
They are recalled in the following table \\
\vspace{-0.25cm}
\begin{center}
\begin{tabular}{|c|c|c|c|c|c|c|}
\hline
n & k & m & q & d  & security\\
\hline
74 & 37 & 41 & 2 & 4 & 80\\
\hline
94 & 47 & 47 & 2 & 5 & 128\\
\hline
68 & 34 & 23 & $2^4$ & 4 & 100\\
\hline
\end{tabular}
\end{center}
 In each case the factorization of $X^k-1$ in $\Fq[X]$  is given by\\
(i)  $X^{37} - 1 = (X-1)\sum_{i=0}^{36} X^i$\\
(ii) $X^{47}-1 = (X-1)PQ$ with $\deg P = \deg Q = 23$\\
(iii) $X^{34}-1 = (X-1)^2(P_1\dots P_8)^2$ with $\deg P_i = 2$, for all $i \in \intvl{1;8}$.\\

In the first case, the polynomial $X^{37} -1$ has only two divisors, so we can only use the first attack.
In the second case, we can choose $D = P$ or $Q$ to obtain a folded code of dimension 23. According to Theorem \ref{th:comp_transposed},  the  complexity of the attack is $\OO{23^3 47^3 2^{3\times 22)}} \approx 2^{96.2}$, that is  a gain around $2^{32}$ compared to the best attack considered in \cite{GMRZ13} and about $2^{20}$ compared to the best attack found in \cite[Subsec. 3.2]{L14}.

The third case is the most interesting. Here we can freely choose the dimension of the projected code. Keep in mind that we want the Gilbert-Varshamov bound greater than $d$  which is the case  when the dimension $k"$ of the projected code is $\geqslant 4$. We choose $k"=4$ and we have in this case an attack of 
complexity $2^{43.6}$ which clearly leads to a feasible attack.

In \cite{GRSZ14}, a new set of parameters is proposed, as follows :\\
\vspace{-0.25cm}
\begin{center}
\begin{tabular}{|c|c|c|c|c|c|c|}
\hline
n & k & m & q & d  & security\\
\hline
82 & 41 & 41 & 2 & 5 & 80\\
\hline
106 & 53 & 53 & 2 & 6 & 128\\
\hline
74 & 37 & 23 & $2^4$ & 4 & 100\\
\hline
\end{tabular}
\end{center}
In each case the factorization of $X^k-1$ in $\Fq[X]$  is given by\\
(i)  $X^{41} - 1 = (X-1)PQ$ with $\deg P = \deg Q = 20$\\
(ii) $X^{53}-1 = (X-1)\displaystyle{\sum_{i=0}^{52}X^i}$\\
(iii) $X^{37}-1 = (X-1) P_1\dots P_4$ with $\deg P_i = 9$, for all $i \in \intvl{1;4}$.\\

The first case allow a non-trivial projection but it is not sufficient to obtain a better complexity than 80. In the second case, we can only use the folding attack, and its complexity is greater than 128.

In the third case, we can choose a projected code of dimension $9$ and we have an attack of complexity $2^{87,1}$, that is a gain around $2^{13}$.

As we can see, it is crucial to choose $k$ such that $X^k-1$ has the minimum of factors in $\Fq[X]$, it can always be factorisable in $(X-1)(\displaystyle{\sum_{i=0}^{k-1} X^i})$ so one have to choose $\displaystyle{\sum_{i=0}^{k-1} X^i}$ irreducible in $\Fq$. This implies $k$ prime but it is not sufficient, as the third case of the parameters proposed in \cite{GRSZ14} proves it.

\bibliography{codecrypto}
\bibliographystyle{plain}

\end{document}